\documentclass[final,11pt]{article}
\pdfoutput=1

\usepackage[thmstyles]{andrews}
\usepackage{andrews_bib}
\usepackage[inline]{showlabels}

\usepackage[vcentermath]{youngtab}

\usepackage[style=alphabetic, 
backref, 
url=false,
doi=false,
isbn=false,
useprefix=true,
maxbibnames=99,
maxcitenames=5,
mincitenames=3,
maxalphanames=5,
minalphanames=3]{biblatex}
\usepackage{changepage}

\usepackage[displaymath, mathlines]{lineno}

\newcommand{\detideal}[3]{I^{\mathrm{det}}_{#1,#2,#3}}

\DeclareMathOperator{\Coeff}{Coeff}

\DeclareMathOperator{\trank}{R}
\DeclareMathOperator{\brank}{\underline{R}}
\DeclareMathOperator{\bmult}{\underline{L}}

\newbibmacro{string+doi}[1]{%
  \iffieldundef{doi}{#1}{\href{http://dx.doi.org/\thefield{doi}}{#1}}}
\DeclareFieldFormat{title}{\usebibmacro{string+doi}{\mkbibemph{#1}}}
\DeclareFieldFormat*{title}{\usebibmacro{string+doi}{\mkbibquote{#1}}}

\title{On Matrix Multiplication and Polynomial Identity Testing}

\author{Robert Andrews\thanks{Department of Computer Science, University of Illinois Urbana-Champaign. Email: \texttt{rgandre2@illinois.edu}. Supported by NSF grant CAREER 20-47310.}}
\date{August 1, 2022}

\begin{document}

\maketitle

\begin{abstract}
	We show that lower bounds on the border rank of matrix multiplication can be used to non-trivially derandomize polynomial identity testing for small algebraic circuits.
	Letting $\brank(n)$ denote the border rank of $n \times n \times n$ matrix multiplication, we construct a hitting set generator with seed length $O(\sqrt{n} \cdot \brank^{-1}(s))$ that hits $n$-variate circuits of multiplicative complexity $s$.
	If the matrix multiplication exponent $\omega$ is not 2, our generator has seed length $O(n^{1 - \eps})$ and hits circuits of size $O(n^{1 + \delta})$ for sufficiently small $\eps, \delta > 0$.
	Surprisingly, the fact that $\brank(n) \ge n^2$ already yields new, non-trivial hitting set generators for circuits of sublinear multiplicative complexity.
\end{abstract}

\section{Introduction}

Matrix multiplication is a fundamental algorithmic problem in theoretical computer science.
Starting with the work of \textcite{Strassen69}, who gave an algorithm to multiply two $n \times n$ matrices in $O(n^{\log_2 7})$ time, a long line of work \cite{Pan78, BCRL79, Pan80, Schonhage81, Rom82, CW82, Strassen87, CW90, DS13, Vas12, LG14, CU03, CKSU05, CU13} has produced faster algorithms to multiply matrices.
Progress on this task is usually measured by $\omega$, the exponent of matrix multiplication, which is the smallest real number such that matrix multiplication can be performed using $O(n^{\omega + \eps})$ arithmetic operations for any positive constant $\eps > 0$.
It is evident that $2 \le \omega \le 3$.
Strassen's \cite{Strassen69} result can be rephrased as a proof that $\omega \le \log_2 7$.
The present state-of-the-art algorithm for matrix multiplication is due to \textcite{AV21}, who proved $\omega < 2.37286$.
It is a major open question to determine whether or not $\omega = 2$.

The complexity of matrix multiplication governs (and in many cases, is equivalent to) the complexity of numerous problems in linear algebra, including computing the determinant and solving systems of linear equations \cite{Strassen69}, boolean matrix multiplication \cite{FM71}, QR decomposition \cite{Schonhage73}, LUP decomposition \cite{BH74}, and computing the coefficients of the characteristic polynomial of a matrix \cite{KG85}.
Fast matrix multiplication has also been used to design algorithms for a host of problems in other areas; examples include recognizing context-free languages \cite{Valiant75}, detecting $k$-cliques \cite{NP85}, and solving linear programs \cite{CLS21,Brand20,JSWZ21}.

While it is popularly conjectured that $\omega = 2$, progress on obtaining improved upper bounds on $\omega$ has slowed over time.
In the three decades since \textcite{CW90} showed $\omega < 2.3755$, the best-known bound on $\omega$ has improved by only $\approx 0.00264$.
The improvements obtained since then \cite{DS13,Vas12,LG14,AV21} apply Strassen's so-called laser method \cite{Strassen87} to powers of the Coppersmith--Winograd tensor.
Recent work \cite{AFG15,BCCGNSU17,BCCGU17,AV18a,AV18b,Alman21,CVZ19,CGLZ20} has shown that this slow progress is no coincidence: there are unconditional barriers to obtaining improved bounds on $\omega$ using generalizations of this and related techniques.

There is a dual line of work concerned with proving lower bounds on the complexity of matrix multiplication.
This usually proceeds by proving lower bounds on the rank or border rank of matrix multiplication, which essentially correspond to the number of multiplications one needs to perform in order to compute a matrix product.
It is known that $\omega = 2$ if and only if the rank (or border rank) of matrix multiplication is bounded from above by $n^{2 + o(1)}$.
The best-known lower bound on the rank of matrix multiplication is $\frac{5}{2}n^2 - 3n$ by \textcite{Blaser99}, with an improvement over finite fields due to \textcite{Shpilka03}.
For border rank, an approximate version of rank, the current record is a lower bound of $2 n^2 - \log_2(n) - 1$, due to \textcite{LM18}.
In a somewhat different vein, \textcite{Raz03} showed that any bounded-coefficient circuit computing $n \times n \times n$ matrix multiplication must be of size $\Omega(n^2 \log n)$.

Naturally, if $\omega = 2$, one obtains extremely fast algorithms for matrix multiplication, leading to improved algorithms for a variety of problems.
However, it is not clear if there is a useful algorithmic consequence of the hypothesis $\omega > 2$.
The main contribution of this work is an application of the assumption $\omega > 2$ to the design of algorithms.
Specifically, we show that if $\omega > 2$, then one can non-trivially derandomize polynomial identity testing for small circuits.

Polynomial identity testing (PIT) is the problem of testing whether an algebraic circuit computes the zero polynomial.
There is a simple, fast randomized algorithm for PIT \cite{Schwartz80,Zippel79}, but no non-trivial deterministic algorithm is known.
Designing a deterministic polynomial-time algorithm for PIT is a major goal of algebraic complexity.
Typically, this is done by constructing a hitting set generator, the analogue of a pseudorandom generator in this setting.
There has been considerable success in derandomizing PIT for restricted classes of circuits (see, e.g., \textcite{SY10,Saxena09,Saxena14}).
For strong models of computation, like formulas and circuits, only conditional results in the form of hardness-to-pseudorandomness results are known \cite{KI04,DSY09,CKS19a,GKSS22,Andrews20,AF22} (see also \cite{KS19}).

\subsection{Our Results} \label{subsec:results}

We now describe our result in more detail.
We construct a hitting set generator for algebraic circuits that have a small number of multiplication gates.

\begin{theorem*}[see \autoref{thm:hsg}]
	Let $\brank(n)$ denote the border rank of $n \times n \times n$ matrix multiplication.
	There is an explicit hitting set generator of seed length $O(\sqrt{n} \brank^{-1}(s))$ that hits $n$-variate circuits with $s$ multiplication gates.
\end{theorem*}

In terms of the matrix multiplication exponent $\omega$, our generator has seed length $O(\sqrt{n} s^{1/\omega})$.
Thus, if $\omega > 2$, we obtain a generator of seed length $O(n^{1 - \eps})$ that hits circuits of size $O(n^{1 + \delta})$ for sufficiently small $\eps, \delta > 0$.
Alternatively, one can phrase this as a win-win result: either $\omega = 2$, giving us fast algorithms for a large collection of problems; or $\omega > 2$, in which case we obtain a non-trivial deterministic algorithm for testing identities given by small circuits.

As $\brank(n) \ge (2 - o(1))n^2$ \cite{LM18}, this also yields an unconditional construction of a hitting set generator with seed length $O(\sqrt{ns})$, which is non-trivial as long as $s \le \eps n$ for sufficiently small $\eps > 0$.
It may seem strange to consider circuits of complexity much less than $n$; for many circuit classes, such circuits are not even capable of reading their entire input.
However, circuits with few multiplication gates are capable of computing non-trivial polynomials, mainly through the use of repeated squaring.
For example, the polynomial $(x_1 + \cdots + x_n)^d$ can be computed using only $O(\log d)$ multiplication gates.

To the best of our knowledge, nothing is known about derandomizing PIT for circuits with few product gates.
For very small $s$, one can obtain non-trivial algorithms by bounding the sparsity of the computed polynomial and using the Klivans--Spielman generator \cite{KS01}.
This strategy breaks down when $s \ge \Omega(\log n)$, as the resulting sparsity bound becomes too large.
In contrast, our construction gives a non-trivial algorithm even when $s = \eps n$ for $\eps < \frac{1}{192}$.

Our result comes within a logarithmic factor of converting all known unconditional hardness for algebraic circuits into pseudorandomness.
The state-of-the-art in explicit lower bounds on multiplicative complexity dates back to \textcite{BS83}, who showed that the polynomial $x_1^d + \cdots + x_n^d$ requires $\Omega(n \log d)$ multiplications to compute.
If one could construct an explicit generator whose seed length remains non-trivial for multiplicative complexity $s \ge \omega(n \log n)$, then this would provide an explicit family of $n$-variate multilinear polynomials of multiplicative complexity $\omega(n \log n)$.
Of course, it remains a possibility that our generator could be improved to hit circuits of multiplicative complexity $O(n \log n)$ without requiring a breakthrough in circuit lower bounds.

As mentioned earlier, there is a collection of works on the hardness-randomness phenomenon in algebraic complexity \cite{KI04,DSY09,CKS19a,GKSS22,Andrews20,AF22}.
Because the assumption $\omega > 2$ is inherently a circuit lower bound, it seems reasonable to expect that one could instantiate the hardness-randomness connection in order to directly obtain our result.
Though this is the spirit of our approach, we remark that the known hardness-randomness framework typically incurs some polynomial overhead in translating a circuit lower bound into a hitting set generator.
In particular, for a weak lower bound of the form $\Omega(n^{1 + \eps})$ (like what is implied by $\omega > 2$), these techniques fail to imply any kind of non-trivial derandomization.
We note that work by \textcite{DST21} showed that for a particular class of constant-variate circuits, weak lower bounds can in fact be used to derandomize PIT.
For more on algebraic hardness versus randomness, see the survey of \textcite{KS19}.

\subsection{Our Techniques}

We briefly describe our generator and the proof of its correctness.
Throughout, we consider $n$-variate circuits as taking as input a matrix $X$ of size $\sqrt{n} \times \sqrt{n}$.
Let $\brank(n)$ be the border rank of $n \times n \times n$ matrix multiplication.
To construct our generator, we will show that the set of matrices of rank $O(\brank^{-1}(s))$ are a hitting set for circuits with $s$ multiplication gates.
This will imply that such a circuit cannot vanish on the product of an $\sqrt{n} \times O(\brank^{-1}(s))$ matrix and an $O(\brank^{-1}(s)) \times \sqrt{n}$ matrix, which yields our generator.
Thus, we are faced with the task of showing that no small circuit can vanish on the set of all matrices of rank $O(\brank^{-1}(s))$.

Let $r \in \naturals$ and let $I_r \subseteq \F[X]$ be the ideal of $\F[X]$ generated by the $r \times r$ minors of the matrix $X$.
It is well-known that when the field $\F$ is algebraically closed, the ideal $I_r$ consists exactly of those polynomials that vanish on matrices of rank less than $r$.
Rephrasing our goal, we need to prove a lower bound of $\Omega(\brank(r))$ on the number of multiplication gates needed to compute any nonzero polynomial in the ideal $I_r$.

Using an observation due to \textcite[Corollary 6]{BS83}, a lower bound on the border rank of matrix multiplication lifts to a lower bound on the border multiplicative complexity of the polynomial $\tr(XYZ)$, where $X$, $Y$, and $Z$ are $n \times n$ matrices.
Results of \textcite{AF22} allow us to further lift this lower bound to the ideal $I_r$ where $r$ is the size of the smallest algebraic branching program computing $\tr(XYZ)$.
Because $\tr(XYZ)$ can be computed by an algebraic branching program with $O(n^2)$ vertices, we obtain a lower bound of $\Omega(\brank(\sqrt{r}))$ on the multiplicative complexity of $I_r$.
This suffices to obtain a hitting set generator of seed length $O(\sqrt{n} \brank^{-1}(s^2))$ for circuits with $s$ product gates.
However, such a construction cannot hope to obtain seed length $o(n)$ for circuits with $O(n^{0.6})$ product gates, even if the best-known upper bound on $\omega$ is tight.

To improve the dependence on $s$ in the seed length, we instead lift the lower bound to the ideal $I_r$ where $r$ is the size of the smallest \emph{trace} algebraic branching program that computes $\tr(XYZ)$.
This polynomial can naturally be computed by a trace ABP of size $O(n)$, which leads to the improved lower bound of $\Omega(\brank(r))$ on the multiplicative complexity of $I_r$.
This immediately translates into the improved seed length of $O(\sqrt{n} \brank^{-1}(s))$ for our hitting set generator.

To perform this improved lifting step, we essentially need to show that trace ABPs of size $s$ can be expressed as a determinant of size $O(s)$.
We do this using the interpretation of the determinant as a sum of weighted cycle covers in an ABP, following \textcite{Valiant79}.

\section{Preliminaries}

Throughout this work, we take $\F$ to be a field of characteristic zero.
For $n \in \naturals$ a natural number, we write $[n] \coloneqq \set{1,2,\ldots,n}$.
We denote by $\vec{x} = (x_1,\ldots,x_n)$ and $X = (x_{i,j})_{i \in [n], j \in [m]}$ a vector of variables and an $n \times m$ matrix of variables, respectively.
We write $\F[\vec{x}]$ for the polynomial ring in the variables $\vec{x}$.
We use $\detideal{n}{m}{r}$ to denote the ideal of $\F[X]$ generated by the $r \times r$ minors of a matrix of variables $X$.
For an $n \times m$ matrix $A$ and subsets $R \subseteq [n]$ and $C \subseteq [m]$, we write $A_{R,C}$ for the submatrix of $A$ obtained by selecting the rows indexed by $R$ and the columns indexed by $C$.

\subsection{Algebraic Circuits}

We briefly recall the notions of algebraic circuits, algebraic branching programs, and trace algebraic branching programs.
For a more thorough treatment of algebraic circuit complexity, we refer the reader to \textcite{SY10} and \textcite{Saptharishi19}.
We begin with the definition of an algebraic circuit.

\begin{definition}
	An \emph{algebraic circuit} is a directed acyclic graph in which every vertex has in-degree zero or two.
	Vertices of in-degree zero are called input gates and are labeled by either a field constant or a variable $x_{i,j}$.
	Vertices of in-degree two are called internal gates and are labeled either as addition or multiplication gates.
	The gates of the circuit compute polynomials in $\F[X]$ in the natural way.
	We allow each edge $e$ of the circuit to be labeled by a field constant $\alpha_e \in \F$, which has the effect of multiplying the value carried by that edge by $\alpha_e$.
	We measure the \emph{size} of a circuit by the number of gates appearing in the circuit.
	The \emph{multiplicative complexity} of a circuit is the number of multiplication gates appearing in the circuit.
\end{definition}

We will also require the notions of algebraic branching programs (ABPs) and trace algebraic branching programs (trace ABPs).

\begin{definition}
	A \emph{(single-source, single-sink) algebraic branching program (ABP)} is a layered directed acyclic graph $G = (V,E)$ with a single source vertex $s$ and a single sink vertex $t$.
	By layered, we mean that there is a partition $V = V_0 \sqcup V_1 \sqcup \cdots \sqcup V_d$ such that $V_0 = \set{s}$, $V_d = \set{t}$, and every edge in $G$ goes from layer $V_{i-1}$ to $V_{i}$ for some $i \in [d]$.
	Every edge $e$ of $G$ is labeled by a linear polynomial $\ell_e(\vec{x}) \in \F[\vec{x}]$.
	Let $\mathcal{P}_{s,t}$ be the set of $s$-$t$ paths in $G$.
	The ABP computes the polynomial given by
	\[
		\sum_{P \in \mathcal{P}_{s,t}} \prod_{e \in P} \ell_e(\vec{x}).
	\]
	The \emph{size} of the ABP is $|V|$, the number of vertices in $G$.
	The \emph{width} of the ABP is $\max_{i \in [d]} |V_i|$.
	
	Equivalently, an ABP is given by a collection of matrices $M_1(\vec{x}),\ldots,M_d(\vec{x})$ whose entries are linear polynomials in $\F[\vec{x}]$.
	The polynomial computed by the ABP is the $(1,1)$ entry of the matrix product $M_1(\vec{x}) \cdots M_d(\vec{x})$, where the dimensions of the matrices $M_i(\vec{x})$ are such that the resulting product is defined.
\end{definition}

A trace ABP endows an ABP with multiple sources $s_1,\ldots,s_m$ and sinks $t_1,\ldots,t_m$.
Whereas an ABP computes a sum over all source-to-sink paths, a trace ABP sums over all $s_i$-$t_i$ paths for all choices of $i \in [m]$, allowing the ABP to reuse intermediate vertices for these different sums.
Alternatively, when viewing an ABP as a matrix product, a trace ABP corresponds to taking the trace of the resulting matrix product instead of extracting the $(1,1)$ entry.

\begin{definition}
	A \emph{trace algebraic branching program (trace ABP)} is a layered directed acyclic graph $G = (V,E)$ with source vertices $s_1,\ldots,s_m$ and sink vertices $t_1,\ldots,t_m$.
	By layered, we mean that there is a partition $V = V_0 \sqcup V_1 \sqcup \cdots \sqcup V_d$ such that $V_0 = \set{s_1,\ldots,s_m}$, $V_d = \set{t_1,\ldots,t_m}$, and every edge in $G$ goes from layer $V_{i-1}$ to $V_{i}$ for some $i \in [d]$.
	Every edge $e$ of $G$ is labeled by a linear polynomial $\ell_e(\vec{x}) \in \F[\vec{x}]$.
	Let $\mathcal{P}_{s_i,t_i}$ be the set of $s_i$-$t_i$ paths in $G$.
	The trace ABP computes the polynomial given by 
	\[
		\sum_{i=1}^m \sum_{P \in \mathcal{P}_{s_i,t_i}} \prod_{e \in P} \ell_e(\vec{x}).
	\]
	The \emph{size} of the ABP is $|V|$, the number of vertices in $G$.
	The \emph{width} of the ABP is $\max_{i \in [d]} |V_i|$.

	Equivalently, a trace ABP is given by a collection of matrices $M_1(\vec{x}),\ldots,M_d(\vec{x})$ whose entries are linear polynomials in $\F[\vec{x}]$.
	The polynomial computed by the trace ABP is the trace of the matrix product $M_1(\vec{x}) \cdots M_d(\vec{x})$, where the dimensions of the matrices $M_i(\vec{x})$ are such that the resulting product is defined.
\end{definition}

It is clear that any polynomial computed by an ABP can be computed by a trace ABP of the same size and width.
Conversely, one can transform a trace ABP into a single-source, single-sink ABP by duplicating the trace ABP $m$ times, deleting all but one pair of source and sink vertices in each copy, and identifying the source vertices and sink vertices in the resulting copies.
To the best of our knowledge, this is the best-known simulation of trace ABPs by single-source, single-sink ABPs.

\begin{lemma} \label{lem:abp simulation}
	Let $f(\vec{x}) \in \F[\vec{x}]$ be a polynomial computed by a trace ABP of size $s$ and width $w$.
	Then $f(\vec{x})$ can be computed by a single-source, single-sink ABP of size $ws$ and width $w^2$.
\end{lemma}

We will make use of the following result of \textcite{BS83} that transforms a circuit that computes a polynomial $f(\vec{x})$ into one that computes all first-order partial derivatives of $f(\vec{x})$ while increasing the circuit size by only a constant factor.
We state the version of their result for multiplicative complexity, although an analogous statement holds for circuit size.
Note that by taking $\F = \mathbb{K}(\eps)$ where $\mathbb{K}$ is a field, this lemma extends to the setting of border complexity (defined in \autoref{subsec:border complexity}).

\begin{lemma}[\cite{BS83}] \label{lem:baur-strassen}
	Let $f(\vec{x}) \in \F[\vec{x}]$ be a polynomial computed by an algebraic circuit of multiplicative complexity $s$.
	Then there is a multi-output algebraic circuit of multiplicative complexity $3s$ that computes $\Set{f(\vec{x}), \frac{\partial f}{\partial x_1}(\vec{x}), \ldots \frac{\partial f}{\partial x_n}(\vec{x})}$.
\end{lemma}

\subsection{Border Complexity} \label{subsec:border complexity}

We will crucially make use of border complexity, which is an approximative version of algebraic computation.

\begin{definition}
	Let $\F$ be a field and $\eps$ be an indeterminate.
	Let $f(\vec{x}) \in \F[\vec{x}]$ be a polynomial.
	We say that a circuit $\Phi$ over the field $\F(\eps)$ \emph{border computes} $f(\vec{x})$ if $\Phi$ computes a polynomial of the form
	\[
		f(\vec{x}) + \eps \cdot g(\vec{x},\eps),
	\]
	where $g(\vec{x},\eps) \in \F[\vec{x},\eps]$.
	We frequently abbreviate this by saying that $\Phi$ computes $f(\vec{x}) + O(\eps)$.
\end{definition}

Over fields of characteristic zero, one can think of border computation as computing a polynomial $f$ up to an arbitrarily small error $\eps$.
The definition above extends to fields of positive characteristic, although this will not be relevant for our work.
Naturally, one can consider the notion of border complexity for restricted classes of circuits, like formulas or branching programs.

If $\mathcal{C}$ is a class of circuits, we define the \emph{closure} of $\mathcal{C}$ to be the set of polynomials that can be border computed by a $\mathcal{C}$-circuit.
For example, if $\mathcal{C}$ is the class of size-$s$ circuits, the closure of $\mathcal{C}$ consists of all polynomials $f(\vec{x})$ such that $f(\vec{x}) + O(\eps)$ can be computed by a size-$s$ circuit over $\F(\eps)$.

In the course of our work, we will prove lower bounds by constructing oracle circuits.
The following lemma says that in the setting of border complexity, one can replace an exact oracle with an approximate oracle without incurring an increase in circuit size.
This makes our job easier, as we only need to reason about circuits using exact oracles.
This lemma is a straightforward consequence of \cite[Lemma 2.3(1)]{Burgisser04}; for a proof, see, e.g., \cite[Lemma 2.3]{AF22}.

\begin{lemma} \label{lem:exact to approx oracle}
	Let $f(\vec{x}), g(\vec{x}) \in \F[\vec{x}]$ be polynomials.
	Suppose $f(\vec{x}) + O(\eps)$ can be computed by a circuit of size $s$ with $g$-oracle gates.
	Let $h(\vec{x},\delta) \in \F\llb \delta \rrb [\vec{x}]$ be a polynomial such that $h(\vec{x},\delta) = g(\vec{x}) + O(\delta)$.
	Then there is some $N \in \naturals$ such that $f(\vec{x}) + O(\eps)$ can be computed by a circuit of size $s$ with $h(\vec{x},\eps^N)$-oracle gates.
\end{lemma}

\subsection{Polynomial Identity Testing}

We will design polynomial identity testing algorithms that operate on circuits in a black-box manner; that is, our algorithms will only evaluate the circuit and will not examine the internal structure of the circuit.
This is equivalent to giving an explicit construction of a hitting set for the class of circuits under consideration.

\begin{definition}
	Let $\mathcal{C} \subseteq \F[\vec{x}]$ be a set of polynomials.
	A set $\mathcal{H} \subseteq \F^n$ is a \emph{hitting set for $\mathcal{C}$} if for every nonzero $f \in \mathcal{C}$, there is some $\vec{\alpha} \in \mathcal{H}$ such that $f(\vec{\alpha}) \neq 0$.
\end{definition}

Equivalently, one can attempt to construct a hitting set generator, which is analogous to a pseudorandom generator in this setting.

\begin{definition}
	Let $\mathcal{C} \subseteq \F[\vec{x}]$ be a set of polynomials.
	A polynomial map $\mathcal{G} : \F^\ell \to \F^n$ is a \emph{hitting set generator for $\mathcal{C}$} if for every nonzero $f \in \mathcal{C}$, we have $f(\mathcal{G}(\vec{y})) \neq 0$.
	We call $\ell$ the \emph{seed length} of the generator.
	The \emph{degree} of the generator, denoted by $\deg(\mathcal{G})$, is given by $\max_{i \in [n]} \deg(\mathcal{G}_i)$.
\end{definition}

One can translate between hitting sets and hitting set generators using the Schwartz--Zippel lemma \cite{Schwartz80,Zippel79} and polynomial interpolation.
We note that if $\mathcal{C} \subseteq \F[X]$ is a set of degree-$d$ polynomials and $\mathcal{G} : \F^\ell \to \F^n$ is a hitting set generator for $\mathcal{C}$, one obtains a hitting set of size $(d \cdot \deg(\mathcal{G}) + 1)^\ell$.
In contrast, one can always construct a hitting set of size $(d+1)^n$.
Note that a generator with $\deg(\mathcal{G}) \le d^{O(1)}$ and $\ell \le o(n)$ corresponds to a hitting set of size $d^{o(n)}$, which is a super-polynomial improvement over the trivial hitting set of size $(d+1)^n$.

\subsection{Determinantal Ideals and Matrix Rank}

Let $X$ be an $n \times m$ matrix of variables.
We denote by $\detideal{n}{m}{r} \subseteq \F[X]$ the ideal of $\F[X]$ generated by the $r \times r$ minors of $X$.
We make use the following proposition of \textcite{AF22}, which reduces the task of proving lower bounds on all polynomials in $\detideal{n}{m}{r}$ to the task of proving lower bounds on products of minors.
We note that the polynomial $(K_\sigma | K_\sigma)(X)$ appearing in the statement of \cite[Proposition 3.5]{AF22} is exactly the same as the product of determinants that appears in the proposition below.
However, we give a more direct statement of this proposition to avoid the language of bitableaux and bideterminants, which is unnecessary for the results of this work.

\begin{proposition}[{\cite[Proposition 3.5]{AF22}}] \label{prop:reduction to single bideterminant}
	Let $f(X) \in \detideal{n}{m}{r}$ be nonzero.
	There is a collection of $nm$ linearly independent linear functions $\ell_{i,j}(X,\eps) \in \F(\eps)[X]$ indexed by $(i,j) \in [n] \times [m]$, an integer $q \in \integers$, a nonzero $\alpha \in \F$, and natural numbers $\sigma_1, \ldots, \sigma_p$ with $\sigma_1 \ge r$ such that
	\[
		f(\ell_{1,1}(X,\eps),\ldots,\ell_{n,m}(X,\eps)) = \eps^q \alpha \prod_{i=1}^{p} \det_{\sigma_i}(X_{[\sigma_i],[\sigma_i]}) + O(\eps^{q+1}).
	\]
\end{proposition}

It is well-known that when the underlying field $\F$ is algebraically closed, the ideal $\detideal{n}{m}{r}$ consists exactly of those polynomials which vanish on all matrices of rank less than $r$.
In particular, proving a lower bound of $s$ on the complexity of all nonzero polynomials in $\detideal{n}{m}{r}$ equates to proving that every polynomial of complexity less than $s$ cannot vanish on all matrices of rank less than $r$.
There is a natural hitting set generator whose image contains all low-rank matrices.

\begin{construction} \label{cons:matrix generator}
	Let $n, m, r \in \naturals$ with $r \le \min(n,m)$.
	Define the map $\mathcal{G}_{n,m,r} : \F^{n \times r} \times \F^{r \times m} \to \F^{n \times m}$ via
	\[
		\mathcal{G}_{n,m,r}(Y,Z)_{i,j} \coloneqq (YZ)_{i,j}.
	\]
\end{construction}

It is evident from its definition that the generator of \autoref{cons:matrix generator} contains in its image all $n \times m$ matrices of rank at most $r$.
The connection between matrix rank and the ideal $\detideal{n}{m}{r}$ can be used to prove the following lemma.
For the sake of completeness, we provide a proof (the same proof can be found in the discussion preceding \cite[Lemma 2.10]{AF22}).

\begin{lemma} \label{lem:vanish on matrix generator} 
	Let $\F$ be any field and let $n,m,r \in \naturals$ with $r \le \min(n,m)$.
	Let $\detideal{n}{m}{r} \subseteq \F[X]$ denote the ideal of $\F[X]$ generated by the $r \times r$ minors of a generic $n \times m$ matrix $X$ and let $f(X) \in \F[X]$.
	Then $f(\mathcal{G}_{n,m,r-1}(Y,Z)) = 0$ if and only if $f(X) \in \detideal{n}{m}{r}$.
\end{lemma}

\begin{proof}
	If $f(X) \in \detideal{n}{m}{r}$, then we can write $f$ as $f(X) = \sum_{i=1}^N g_i(X) h_i(X)$ where the polynomials $\set{g_1,\ldots,g_N}$ are the $r \times r$ minors of $X$.
	Because the image of $\mathcal{G}_{n,m,r-1}(Y,Z)$ is necessarily a matrix of rank at most $r-1$, each $r \times r$ minor of $\mathcal{G}_{n,m,r-1}(Y,Z)$ vanishes, i.e., $g_i(\mathcal{G}_{n,m,r-1}(Y,Z)) = 0$ for all $i \in [N]$.
	This implies $f(\mathcal{G}_{n,m,r-1}(Y,Z)) = 0$.

	To prove the converse direction, we first work under the assumption that the field $\F$ is algebraically closed.
	Suppose that $f(\mathcal{G}_{n,m,r-1}(Y,Z)) = 0$.
	Let $J_{n,m,r-1} \subseteq \F[X]$ be the ideal of $\F[X]$ consisting of polynomials that vanish on the set of matrices of rank at most $r-1$.
	Because the image of $\mathcal{G}_{n,m,r-1}(Y,Z)$ contains all matrices of rank at most $r-1$, we have $f \in J_{n,m,r-1}$.
	To show that $f \in \detideal{n}{m}{r}$, we will prove the equality $\detideal{n}{m}{r} = J_{n,m,r-1}$.

	The inclusion $\detideal{n}{m}{r} \subseteq J_{n,m,r-1}$ is immediate, as the $r \times r$ minors vanish on matrices of rank less than $r$.
	For the inclusion in the reverse direction, we use the correspondence between ideals and varieties.
	Recall that for an ideal $I \subseteq F[X]$, we denote by $V(I) \subseteq \F^{n \times m}$ the \emph{variety} of $I$, defined as 
	\[
		V(I) \coloneqq \set{A \in \F^{n \times m} : \forall h(X) \in I, h(A) = 0}.
	\]
	Let $V(\detideal{n}{m}{r})$ be the variety over $\F$ defined by the ideal $\detideal{n}{m}{r}$ and let $A \in V(\detideal{n}{m}{r})$ be a point in this variety.
	By definition, each $r \times r$ minor of $A$ vanishes, so $\rank(A) \le r-1$, which implies $A \in V(J_{n,m,r-1})$.
	This shows $V(\detideal{n}{m}{r}) \subseteq V(J_{n,m,r-1})$.
	By Hilbert's Nullstellensatz, this implies $\sqrt{J_{n,m,r-1}} \subseteq \sqrt{\detideal{n}{m}{r}}$, where $\sqrt{I}$ denotes the \emph{radical} of an ideal $I$.
	The ideal $\detideal{n}{m}{r}$ is radical (see, e.g., \cite[Theorem 2.10 and Remark 2.12]{BV88}), so we have the desired inclusion
	\[
		J_{n,m,r-1} \subseteq \sqrt{J_{n,m,r-1}} \subseteq \sqrt{\detideal{n}{m}{r}} = \detideal{n}{m}{r}.
	\]
	This proves $J_{n,m,r-1} = \detideal{n}{m}{r}$, hence $f(X) \in \detideal{n}{m}{r}$ as claimed.

	If $\F$ is not algebraically closed, we can still consider $f(X)$ as a polynomial over the algebraic closure $\overline{\F}$.
	If $f(\mathcal{G}_{n,m,r-1}(Y,Z)) = 0$, the previous argument implies that $f \in \detideal{n}{m}{r}$ when $\detideal{n}{m}{r}$ is considered as an ideal over $\overline{\F}$.
	Letting $I_\F$ and $I_{\overline{\F}}$ denote $\detideal{n}{m}{r}$ when considered as an ideal over $\F$ and $\overline{\F}$, respectively, we have $f \in I_{\overline{\F}} \cap \F[\vec{x}]$.
	\autoref{lem:ideal base change} below shows that $I_{\overline{\F}} \cap \F[\vec{x}] = I_{\F}$, so we in fact have $f \in I_\F$ as desired.
\end{proof}

The following is an elementary lemma used in the proof of \autoref{lem:vanish on matrix generator} in the case where $\F$ is not algebraically closed.
In the spirit of keeping this work self-contained, we provide a proof.
\todo{Is there a reference for this lemma?}

\begin{lemma} \label{lem:ideal base change}
	Let $\F$ be a field and let $\K \supseteq \F$ be an extension of $\F$.
	Let $\set{g_1,\ldots,g_m} \subseteq \F[\vec{x}]$ be a set of polynomials.
	Let $I_\F$ and $I_\K$ be the ideals generated by $\set{g_1,\ldots,g_m}$ over $\F[\vec{x}]$ and $\K[\vec{x}]$, respectively.
	Then $I_\F = I_\K \cap \F[\vec{x}]$.
\end{lemma}

\begin{proof}
	The inclusion $I_\F \subseteq I_{\K} \cap \F[X]$ is immediate.
	For the other direction, let $\set{v_1,v_2,\ldots}$ be a basis of $\K$ as a vector space over $\F$ with the additional property that $v_1$ spans $\F$.
	Consider the linear projection $\pi : \K \to \F$ that sends $v_1$ to itself and $v_i$ to zero for $i \ge 2$.
	We extend $\pi$ to a projection $\pi : \K[X] \to \F[X]$ by applying the projection from $\K$ to $\F$ coefficient-wise.
	Let $f(X) \in I_{\K} \cap \F[X]$ be given by $f(X) = \sum_{i=1}^N g_i(X) h_i(X)$.
	We claim that $f(X) = \sum_{i=1}^N g_i(X) \pi(h_i(X))$, which proves $f(X) \in I_\F$ as desired.

	To see this, let $m$ be a monomial and consider the coefficient $\Coeff_m(f)$ of $m$ in $f$.
	Because $\Coeff_m(f) \in \F[X]$, we have $\pi(\Coeff_m(f)) = \Coeff_m(f)$.
	Using the fact that $\pi$ is $\F$-linear, this implies
	\begin{align*}
		\Coeff_m(f) &= \pi(\Coeff_m(f)) \\
		&= \sum_{i=1}^N \pi(\Coeff_m(g_i h_i)) \\
		&= \sum_{i=1}^N \sum_{m = m_1 m_2} \pi(\Coeff_{m_1}(g_i) \Coeff_{m_2}(h_i)) \\
		&= \sum_{i=1}^N \sum_{m = m_1 m_2} \Coeff_{m_1}(g_i) \pi(\Coeff_{m_2}(h_i)) \\
		&= \sum_{i=1}^N \Coeff_m(g_i \pi(h_i)),
	\end{align*}
	where the inner sum is over all monomials $m_1$ and $m_2$ whose product is $m$.
	The equality 
	\[
		\pi(\Coeff_{m_1}(g_i) \Coeff_{m_2}(h_i)) = \Coeff_{m_1}(g_i) \pi(\Coeff_{m_2}(h_i))
	\]
	follows from the fact that $\pi$ is $\F$-linear and $\Coeff_{m_1}(g_i) \in \F$.
	Thus, $f(X) = \sum_{i=1}^N g_i(X) \pi(h_i(X))$.
\end{proof}

One can use \autoref{lem:vanish on matrix generator} to design PIT algorithms for circuit classes $\mathcal{C}$ that are too weak to efficiently compute a nonzero element of $\detideal{n}{m}{r}$.
If every small $\mathcal{C}$-circuit cannot compute a nonzero element of $\detideal{n}{m}{r}$, then \autoref{lem:vanish on matrix generator} implies that the map $\mathcal{G}_{n,m,r-1}(Y,Z)$ of \autoref{cons:matrix generator} is a hitting set generator for the class of small $\mathcal{C}$-circuits.

\subsection{Complexity of Matrix Multiplication}

This subsection introduces the language of tensors and their relationship with the complexity of matrix multiplication.
For a more thorough treatment of tensors and matrix multiplication, we refer the reader to \textcite[Chapters 14 and 15]{BCS97} and \textcite{Blaser13}.

For our purposes, a tensor $T$ of order $d$ is a set-multilinear polynomial in $d$ disjoint sets of variables $X^{(1)}, \ldots, X^{(d)}$.
The fact that $T$ is set-multilinear means that every monomial appearing in $T$ is a product of $d$ variables, where exactly one of these variables is taken from each of the sets $X^{(1)},\ldots,X^{(d)}$.
That is, we can write $T$ as
\[
	T(X^{(1)},\ldots,X^{(d)}) = \sum_{i_1=1}^{n_1} \cdots \sum_{i_d=1}^{n_d} t_{i_1,\ldots,i_d} x^{(1)}_{i_1} \cdots x^{(d)}_{i_d}.
\]

We say that a tensor is rank-one if there are linear forms $\ell_1(X^{(1)}), \ldots, \ell_d(X^{(d)})$ such that
\[
	T(X^{(1)}, \ldots, X^{(d)}) = \ell_1(X^{(1)}) \cdots \ell_d(X^{(d)}).
\]
The rank of $T$, written as $\trank(T)$, is the minimal $r$ such that $T$ can be written as a sum of rank-one tensors.
The border rank of $T$, denoted by $\brank(T)$, is the minimal $r$ such that $T$ can be obtained as a limit of rank-$r$ tensors.
More explicitly, a tensor $T$ has border rank $r$ if there are linear forms $\ell_{i,j}(X^{(i)},\eps) \in \F(\eps)[X^{(i)}]$ such that
\[
	\sum_{j=1}^r \prod_{i=1}^d \ell_{i,j}(X^{(i)},\eps) = T(X^{(1)},\ldots,X^{(d)}) + O(\eps)
\]
and there is no such expression for $T + O(\eps)$ involving fewer than $r$ rank-one tensors.

We denote by $\abr{n,m,p}$ the order-3 tensor
\[
	\abr{n,m,p} \coloneqq \sum_{i=1}^n \sum_{j=1}^m \sum_{k=1}^p x_{i,j} y_{j,k} z_{i,k},
\]
which corresponds to the multiplication of an $n \times m$ matrix with an $m \times p$ matrix.
Note that $\abr{n,m,p} = \tr(XYZ^\top)$, a fact that we will use later on.

The complexity of $n \times n \times n$ matrix multiplication is captured by the rank of the tensor $\abr{n,n,n}$ (see, e.g., \cite[Proposition 15.1]{BCS97}).
We now define $\omega$, the exponent of matrix multiplication.

\begin{definition}
	$\omega \coloneqq \inf \set{\tau \in \reals : \trank(\abr{n,n,n}) \le O(n^{\tau})}$.
\end{definition}

\textcite{Bini80} showed that one can equivalently define $\omega$ in terms of the border rank of $\abr{n,n,n}$.

\begin{lemma}[\cite{Bini80}]
	$\omega = \inf \set{\tau \in \reals : \brank(\abr{n,n,n}) \le O(n^{\tau})}$.
\end{lemma}

As mentioned in the introduction, the obvious bounds on $\omega$ are $2 \le \omega \le 3$.
The best-known upper bound on $\omega$ is due to \textcite{AV21}.

\begin{theorem}[\cite{AV21}]
	$\omega < 2.37286$.
\end{theorem}

It is popularly conjectured that $\omega = 2$.
There has been some progress on lower bounds for $\brank(\abr{n,n,n})$, with the best-known lower bound due to \textcite{LM18}.

\begin{theorem}[\cite{LM18}] \label{thm:brank lb}
	$\brank(\abr{n,n,n}) \ge 2n^2 - \log_2 n - 1$.
\end{theorem}

One can also consider the multiplicative complexity of matrix multiplication, where we do not restrict ourselves to computing variable-disjoint products of the form $\ell_1(X) \ell_2(Y)$, but instead consider products $\ell_1(X,Y) \ell_2(X,Y)$ of arbitrary linear polynomials.
The following lemma shows that for matrix multiplication, border rank and border multiplicative complexity differ by at most a factor of 2.
In the case of (exact) rank and multiplicative complexity, this is a well-known fact (see, e.g., \cite[Eqn.~14.8]{BCS97} and the discussion preceding it).
The proof for the case of border computation is nearly identical.

\begin{lemma}[{cf.~\cite[Eqn.~14.8]{BCS97}}] \label{lem:MM border rank vs mult complexity}
	Let $\bmult(n)$ denote the border multiplicative complexity of $n \times n \times n$ matrix multiplication.
	Then $\bmult(n) \le \brank(\abr{n,n,n}) \le 2 \bmult(n)$.
\end{lemma}

\section{Lifting Border Rank Lower Bounds to Determinantal Ideals} \label{sec:trace abp}

In this section, we will show that lower bounds on the border rank of matrix multiplication can be lifted to lower bounds on the border multiplicative complexity of any nonzero polynomial in the ideal $\detideal{n}{m}{r} \subseteq \F[X]$.
Letting $\brank(n) \coloneqq \brank(\abr{n,n,n})$ be the border rank of $n \times n \times n$ matrix multiplication, our goal will be to prove a lower bound of order $\brank(r)$ on the border multiplicative complexity of the ideal $\detideal{n}{m}{r}$.
To do this, we make use of tools recently developed by \textcite{AF22} to prove lower bounds on the complexity of polynomials in this ideal.

We now state and prove our main technical lemma, which is an analogue of \cite[Lemma 3.6]{AF22} for trace ABPs.

\begin{lemma} \label{lem:trace abp projection}
	Let $\F$ be a field of characteristic zero.
	Let $X^{(1)},\ldots,X^{(m)}$ be matrices of variables, where $X^{(i)}$ is an $n_i \times n_{i+1}$ matrix and $n_1 = n_{m+1}$.
	Let $N \coloneqq \sum_{i=1}^{m+1} n_i$.
	Let $\sigma = (\sigma_1,\ldots,\sigma_p)$ be a non-increasing sequence of natural numbers with $\sigma_1 \ge N$.
	Then there is a matrix $M \in \F(\eps)[X^{(1)},\ldots,X^{(m)}]^{\sigma_1 \times \sigma_1}$ where each entry $M_{i,j}$ is either a constant or a scalar multiple of a variable and we have
	\[
		\prod_{i=1}^p \det(M_{[\sigma_i],[\sigma_i]}) = 1 + \eps \tr(X^{(1)} \cdots X^{(m)}) + O(\eps^2).
	\]
\end{lemma}

\begin{proof}
	Without loss of generality, it suffices to consider the case where $\sigma_1 = N$.
	If instead $\sigma_1 > N$, we extend the matrix $M$ to a $\sigma_1 \times \sigma_1$ matrix by placing ones along the main diagonal and zeroes elsewhere.

	Let $G$ be the underlying directed graph of the trace ABP that computes $\tr(X^{(1)} \cdots X^{(m)})$.
	We modify $G$ as follows:
	\begin{itemize}
		\item
			Add a self-loop of weight $1$ to every vertex of $G$.
		\item
			Let $s_1,\ldots,s_{n_1}$ denote the sources of $G$ and $t_1,\ldots,t_{n_1}$ the corresponding sinks. 
			Add an edge of weight $\eps$ from $t_i$ to $s_i$ for every $i \in [n_1]$.
	\end{itemize}
	Let $G'$ denote the resulting graph and let $M'$ be the adjacency matrix of $G'$, i.e.,
	\[
		M' \coloneqq 
		\begin{pmatrix}
			I_{n_1} & X^{(1)} & 0 & 0 & \cdots & 0 & 0 \\
			0 & I_{n_2} & X^{(2)} & 0 & \cdots & 0 & 0 \\
			0 & 0 & I_{n_3} & X^{(3)} & \cdots & 0 & 0 \\
			0 & 0 & 0 & I_{n_4} & \cdots & 0 & 0 \\
			\vdots & \vdots & \vdots & \vdots & \ddots & \vdots & \vdots \\
			0 & 0 & 0 & 0 & \cdots & I_{n_{m-1}} & X^{(m)} \\
			\eps I_{n_1} & 0 & 0 & 0 & \cdots & 0 & I_{n_m} \\
		\end{pmatrix}.
	\]
	We will first determine $\prod_{i=1}^p \det(M'_{[\sigma_i],[\sigma_i]})$, after which we will modify $M'$ to obtain the desired matrix $M$.

	Fix some $k \in [N]$.
	If $k \le \sum_{i=1}^{m-1} n_i$, then it is clear that $\det(M'_{[k],[k]}) = 1$, as $M'_{[k],[k]}$ is an upper triangular matrix with ones along the diagonal.
	For $k$ in the range $N-n_m < k \le N$, we compute $\det(M'_{[k],[k]})$ using the cycle cover interpretation of the determinant.

	Let $G'_k$ denote the graph whose adjacency matrix is $M'_{[k],[k]}$.
	Recall that $\det(M'_{[k],[k]})$ can be computed as
	\[
		\det(M'_{[k],[k]}) = \sum_{C \in \mathscr{C}(G'_k)} (-1)^{\mathrm{even}(C)} \pi(C),
	\]
	where $\mathscr{C}(G'_k)$ is the set of all cycle covers in $G'_k$, $\mathrm{even}(C)$ is the number of even cycles in $C$, and $\pi(C)$ is the product of the weights on the edges appearing in the cycle cover $C$.
	We partition the set of cycle covers of $G'_k$ into three sets: those containing no edges of weight $\eps$, those containing exactly one edge of weight $\eps$, and those containing two or more edges of weight $\eps$.
	In each case, we determine the contribution of these cycle covers to $\det(M'_{[k],[k]})$.

	\begin{itemize}
		\item
			Suppose $C$ is a cycle cover with no edges of weight $\eps$.
			The construction of $G'_k$ implies that $C$ must be the cycle cover consisting entirely of self-loops.
			This cycle cover contributes $1$ to $\det(M'_{[k],[k]})$.

		\item
			Let $C$ be a cycle cover containing exactly one edge labeled $\eps$.
			By the construction of $G'_k$, the cycle in $C$ containing the edge labeled $\eps$ must correspond to a path from $s_i$ to $t_i$ in $G$ for some $i \in [k - (N - n_m)]$ together with the $\eps$ edge from $t_i$ to $s_i$.
			Because every non-trivial cycle in $G'_k$ must use an edge labeled $\eps$, the remaining cycles in $C$ consist of self-loops.
			Thus, $C$ contributes a term of the form
			\[
				(-1)^{m+1} \eps X^{(1)}_{i,i_2} X^{(2)}_{i_2,i_3} \cdots X^{(m-1)}_{i_{m-1},i_m} X^{(m)}_{i_m,i}.
			\]
			to $\det(M'_{[k],[k]})$, where the factor of $(-1)^{m+1}$ accounts for the parity of the length of the non-trivial cycle.
			There is exactly one such cycle cover for every $i \in [k - (N-n_m)]$ and every path from $s_i$ to $t_i$ in $G_k$. 
			This implies that the set of all cycle covers containing exactly one edge of weight $\eps$ contributes
			\begin{multline*}
				(-1)^{m+1} \eps \sum_{i_1=1}^{N-n_m+k} \sum_{i_2,\ldots,i_m} X^{(1)}_{i_1,i_2} X^{(2)}_{i_2,i_3} \cdots X^{(m)}_{i_m,i_1} \\ = (-1)^{m+1} \eps \tr((X^{(1)} \cdots X^{(m)})_{[k-(N-n_m)],[k-(N-n_m)]})
			\end{multline*}
			to $\det(M'_{[k],[k]})$.

		\item
			Finally, consider the case when $C$ is a cycle cover containing two or more edges labeled by $\eps$.
			By definition, this cycle cover contributes an $O(\eps^2)$ term to $\det(M'_{[k],[k]})$, which we consider negligible.
	\end{itemize}

	In summary, we have
	\[
		\det(M'_{[k],[k]}) = 1 + (-1)^{m+1} \eps \tr((X^{(1)} \cdots X^{(m)})_{[k-(N-n_m)], [k-(N-n_m)]}) + O(\eps^2).
	\]
	Using this, we now determine $\prod_{i=1}^p \det(M'_{[\sigma_i],[\sigma_i]})$.
	Let $a_i \coloneqq \Abs{\set{j \in [p] : \sigma_j = i}}$ count the number of elements of $\sigma$ equal to $i$.
	The analysis above implies
	\begin{align*}
		\prod_{i=1}^p \det(M'_{[\sigma_i],[\sigma_i]}) &= \prod_{k=1}^N \det(M'_{[k],[k]})^{a_k} \\
		&= \prod_{\ell=1}^{n_m} \det(M'_{[N-n_m+\ell],[N-n_m+\ell]})^{a_{N-n_m+\ell}} \\ 
		&= \prod_{\ell=1}^{n_m} \del{1 + (-1)^{m+1} \eps \tr((X^{(1)} \cdots X^{(m)})_{[\ell],[\ell]}) + O(\eps^2)}^{a_{N-n_m+\ell}} \\
		&= \prod_{\ell=1}^{n_m} \del{1 + (-1)^{m+1} \eps a_{N - n_m + \ell} \tr((X^{(1)} \cdots X^{(m)})_{[\ell],[\ell]}) + O(\eps^2)} \\
		&= 1 + (-1)^{m+1} \eps \sum_{\ell=1}^{n_m} a_{N-n_m+\ell} \tr((X^{(1)} \cdots X^{(m)})_{[\ell],[\ell]}) + O(\eps^2) \\
		&= 1 + (-1)^{m+1} \eps \sum_{i=1}^{n_m} \sbr{ \del{\sum_{\ell=i}^{n_m} a_{N - n_m + \ell}} (X^{(1)} \cdots X^{(m)})_{i,i}} + O(\eps^2).
	\end{align*}

	We now perform a change of variables to transform the matrix $M'$ into the desired matrix $M$.
	Let $A$ be the diagonal matrix given by
	\[
		A_{i,i} = \frac{1}{\sum_{\ell = i}^{n_m}a_{N - n_m + \ell}}.
	\]
	(Note that the entries of $A$ are well-defined, since $a_N \ge 1$ and $a_i \ge 0$ for all $i \in [N]$.)
	Let $M$ be the image of $M'$ under the change of variables $X^{(1)} \mapsto (-1)^{m+1} A X^{(1)}$.
	Then we have
	\begin{align*}
		\prod_{i=1}^p \det(M_{[\sigma_i],[\sigma_i]}) &= 1 + (-1)^{m+1} \eps \sum_{i=1}^{n_m}\sbr{\del{\sum_{\ell=i}^{n_m} a_{N - n_m + \ell}} ((-1)^{m+1} A X^{(1)} \cdots X^{(m)})_{i,i}} + O(\eps^2) \\
		&= 1 + \eps \sum_{i=1}^{n_m}\sbr{A_{i,i}^{-1} A_{i,i} (X^{(1)} \cdots X^{(m)})_{i,i}} + O(\eps^2) \\
		&= 1 + \eps \sum_{i=1}^{n_m} (X^{(1)} \cdots X^{(m)})_{i,i} + O(\eps^2) \\
		&= 1 + \eps \tr(X^{(1)} \cdots X^{(m)}) + O(\eps^2). \qedhere
	\end{align*}
\end{proof}

\begin{remark}
	In the proof of the preceding lemma, suppose we were to add edges of weight 1 from $t_i$ to $s_i$ for each $i$ and add self-loops of weight 1 to all vertices.
	To compute $\tr(X^{(1)} \cdots X^{(d)})$ using the cycle cover interpretation of the determinant, we want to restrict ourselves to only count cycle covers containing a single edge $t_i$-$s_i$ edge.
	We accomplish this by multiplying the weight of each such edge by a factor of $\eps$, which guarantees that the linear term of the determinant of the adjacency matrix corresponds to cycle covers using exactly one $t_i$-$s_i$ edge.
	In fact, we get more: the coefficient of $\eps^k$ in the determinant of the adjacency matrix corresponds to cycle covers using exactly $k$ such edges.

	A similar idea is used in algorithms for ``exact'' problems in combinatorial optimization.
	For example, the algorithms of \textcite{BP87} for counting exact arborescences and exact perfect matchings in planar graphs modify the edge weights of the graph in a manner similar to what we do in the proof of \autoref{lem:trace abp projection}.
	By exploiting the notion of border complexity, we avoid an interpolation step used in these combinatorial algorithms.
\end{remark}

Using the preceding lemma, we establish an analogue of \cite[Theorem 3.8]{AF22} for trace ABPs.

\begin{proposition} \label{prop:trace abp reduction}
	Let $\F$ be a field of characteristic zero.
	Let $f(X) \in \detideal{n}{m}{r}$ be a nonzero polynomial and let $h(X, \eps) \in \F\llb \eps \rrb [X]$ be any polynomial such that $h(X,\eps) = f(X) + O(\eps)$.
	Let $g(\vec{y}) \in \F[\vec{y}]$ be a polynomial in the border of layered trace algebraic branching programs with at most $r$ vertices.
	Then there is a depth-three $h$-oracle circuit $\Phi$ defined over $\F(\eps)$ such that the following hold.
	\begin{enumerate}
		\item
			$\Phi$ has $nm$ addition gates at the bottom layer, a single $h$-oracle gate in the middle layer, and a single addition gate at the top layer.
		\item
			$\Phi$ computes $g(\vec{y}) + O(\eps)$.
	\end{enumerate}
\end{proposition}

\begin{proof}
	By \autoref{lem:exact to approx oracle}, it is sufficient to consider the case where the oracle gates compute $f(X)$ exactly.
	Using \autoref{prop:reduction to single bideterminant}, there are $nm$ linear functions $\set{\ell_{i,j}(X,\eps) \in \F(\eps)[X] : (i,j) \in [n] \times [m]}$, an integer $q \in \integers$, a nonzero $\alpha \in \F$, and a sequence $\sigma = (\sigma_1,\ldots,\sigma_p)$ of natural numbers with $\sigma_1 \ge r$ such that
	\[
		f(\ell_{1,1}(X,\eps),\ldots,\ell_{n,m}(X,\eps)) = \eps^q \alpha \prod_{i=1}^p \det(X_{[\sigma_i],[\sigma_i]}) + O(\eps^{q+1}).
	\]

	By assumption, there is a polynomial $\tilde{g}(\vec{y},\eps) \in \F(\eps)[\vec{y}]$ such that $\tilde{g}(\vec{y},\eps) = g(\vec{y}) + O(\eps)$ and that $\tilde{g}(\vec{y},\eps)$ can be computed by a layered trace ABP on $s$ vertices for some $s \le r$.
	That is, there are matrices of variables $Z^{(1)}, \ldots, Z^{(m)}$, where $Z^{(i)}$ is an $n_i \times n_{i+1}$ matrix, we have $n_1 = n_{m+1}$, and $\sum_{i=1}^{m+1} = s$, along with a projection $\varphi : Z^{(1)} \cup \cdots \cup Z^{(m)} \to \vec{y} \cup \F(\eps)$ such that $\tr(\varphi(Z^{(1)}) \cdots \varphi(Z^{(m)})) = \tilde{g}(\vec{y},\eps)$.

	Applying \autoref{lem:trace abp projection} to the matrices $Z^{(1)},\ldots,Z^{(m)}$ and the sequence $(\sigma_1,\ldots,\sigma_p)$, we obtain a matrix $M(Z,\eps) \in \F(\eps)[Z^{(1)}, \ldots, Z^{(m)}]^{r \times r}$ such that
	\[
		\prod_{i=1}^p \det(M(Z,\eps)_{[\sigma_i],[\sigma_i]}) = 1 + \eps \tr(Z^{(1)} \cdots Z^{(m)}) + O(\eps^2).
	\]

	We now compose $f(X)$, the linear functions $\ell_{i,j}(X,\eps)$, the matrix $M(Z,\eps)$, and the projection $\varphi : Z \to \vec{y} \cup \F(\eps)$.
	Let
	\[
		h(\vec{y},\eps,\delta) \coloneqq f(\ell_{1,1}(M(\varphi(Z),\delta),\eps), \ldots, \ell_{n,m}(M(\varphi(Z),\delta),\eps)).
	\]
	The preceding discussion implies
	\begin{align*}
		h(\vec{y},\eps,\delta) &= \eps^q \alpha \cdot \prod_{i=1}^p \det(M(\varphi(Z),\delta)_{[\sigma_i],[\sigma_i]}) + O(\eps^{q+1}) \\
		&= \eps^q \alpha \cdot \del{1 + \delta \tr(\varphi(Z^{(1)}) \cdots \varphi(Z^{(m)})) + O(\delta^2)} + O(\eps^{q+1}) \\
		&= \eps^q \alpha + \eps^q \delta \alpha \tilde{g}(\vec{y},\eps) + O(\eps^q \delta^2) + O(\eps^{q+1}).
	\end{align*}
	Performing the substitution $\eps \mapsto \eps^2$ and $\delta \mapsto \eps$, we obtain
	\begin{align*}
		h(\vec{y},\eps^2,\eps) &= \eps^{2q} \alpha + \eps^{2q+1} \alpha \tilde{g}(\vec{y},\eps^2) + O(\eps^{2q+2}) \\
		&= \eps^{2q} \alpha + \eps^{2q+1} \alpha g(\vec{y}) + O(\eps^{2q+2}).
	\end{align*}
	The desired $f$-oracle circuit is then given by
	\[
		\Phi(\vec{y}) \coloneqq \frac{h(\vec{y},\eps^2,\eps) - \eps^{2q} \alpha}{\eps^{2q+1} \alpha} = g(\vec{y}) + O(\eps). \qedhere
	\]
\end{proof}

We now use \autoref{prop:trace abp reduction} to lift lower bounds on the border rank of matrix multiplication to lower bounds on the border multiplicative complexity of the ideal $\detideal{n}{m}{r}$.

\begin{theorem} \label{thm:mult complexity lb}
	Let $\F$ be a field of characteristic zero.
	The border multiplicative complexity of any nonzero polynomial in $\detideal{n}{m}{r}$ is bounded from below by $\frac{1}{6} \brank(r/4)$, where $\brank(n) \coloneqq \brank(\abr{n,n,n})$ is the border rank of $n \times n \times n$ matrix multiplication.
\end{theorem}

\begin{proof}
	Let $\Phi$ be a circuit of border multiplicative complexity $s$ computing a nonzero polynomial in $\detideal{n}{m}{r}$.
	Let $X$, $Y$, and $Z$ be $r/4 \times r/4$ matrices of variables.
	The polynomial $\tr(XYZ)$ can naturally be computed by a layered trace ABP on $r$ vertices.
	Applying \autoref{prop:trace abp reduction} to the circuit $\Phi$ yields a circuit $\Psi$ of multiplicative complexity $s$ that computes $\tr(XYZ) + O(\eps)$.
	We then apply \autoref{lem:baur-strassen} to $\Psi$ to obtain a circuit of multiplicative complexity $3s$ that simultaneously computes all first-order partial derivatives of $\tr(XYZ) + O(\eps)$.

	Observe that the partial derivative of $\tr(XYZ)$ with respect to $z_{j,i}$ is, up to the $O(\eps)$ error term, the $(i,j)$ entry of the matrix product $XY$.
	Thus, we have a circuit of multiplicative complexity $3s$ that approximates the product of two $r/4 \times r/4$ matrices.
	By \autoref{lem:MM border rank vs mult complexity}, this implies that the border rank of $r/4 \times r/4 \times r/4$ matrix multiplication is bounded from above by $6s$.
	That is, we have $\brank(r/4) \le 6s$.
	This yields the claimed lower bound on the multiplicative complexity of any nonzero polynomial in $\detideal{n}{m}{r}$.
\end{proof}

Combining \autoref{thm:brank lb} with \autoref{thm:mult complexity lb} yields the following unconditional lower bound on the border multiplicative complexity of all nonzero polynomials in the ideal $\detideal{n}{m}{r}$.

\begin{corollary}
	The border multiplicative complexity of any nonzero polynomial in $\detideal{n}{m}{r}$ is bounded from below by $\frac{1}{48} r^2 - \frac{1}{6} \log_2 r + \frac{1}{6}$.
\end{corollary}

\section{Constructing a Hitting Set Generator} \label{sec:PIT}

In this section, we use \autoref{thm:mult complexity lb} to design hitting set generators for the closure of circuits of small multiplicative complexity.
Letting $\brank(n) \coloneqq \brank(\abr{n,n,n})$ be the border rank of $n \times n \times n$ matrix multiplication, we will construct a generator with seed length $O(\sqrt{n} \brank^{-1}(s))$ for $n$-variate circuits of multiplicative complexity $s$.
We stress that the correctness of this generator is unconditional.

\begin{theorem} \label{thm:hsg}
	Let $\F$ be a field of characteristic zero.
	Let $\brank(n) \coloneqq \brank(\abr{n,n,n})$ be the border rank of $n \times n \times n$ matrix multiplication.
	Then there is an explicit degree-two hitting set generator of seed length $8 \sqrt{n} \brank^{-1}(6s+1)$ that hits the closure of $n$-variate circuits of multiplicative complexity $s$.
\end{theorem}

\begin{proof}
	Let $\Phi$ be an $n$-variate circuit of multiplicative complexity $s$ that computes $\Phi(\vec{x}) + O(\eps)$ for some nonzero polynomial $\Phi(\vec{x})$.
	Let $r \coloneqq 4 \brank^{-1}(6s+1)$.
	Arrange the input variables of $\Phi(\vec{x})$ into a $\sqrt{n} \times \sqrt{n}$ matrix.
	Let $\mathcal{G}_{n,m,r}(Y,Z)$ be the generator of \autoref{cons:matrix generator}.
	We claim that the generator $\mathcal{G}_{\sqrt{n},\sqrt{n},r-1}(Y,Z)$ hits $\Phi(\vec{x})$, i.e., that $\Phi(\mathcal{G}_{\sqrt{n},\sqrt{n},r-1}(Y,Z)) \neq 0$.

	To see this, suppose instead that $\Phi(\mathcal{G}_{\sqrt{n},\sqrt{n},r-1}(Y,Z)) = 0$.
	\autoref{lem:vanish on matrix generator} implies that $\Phi(\vec{x}) \in \detideal{\sqrt{n}}{\sqrt{n}}{r} \setminus \set{0}$.
	As $\Phi(\vec{x})$ has border multiplicative complexity $s$, it follows from \autoref{thm:mult complexity lb} that $6s \ge \brank(r/4)$.
	However, our choice of $r$ implies $\brank(r/4) = 6s+1 > 6s$, a contradiction.
	Thus, it must be the case that in fact $\mathcal{G}_{\sqrt{n},\sqrt{n},r-1}(Y,Z)$ hits $\Phi(\vec{x})$.
	Since $\Phi$ was an arbitrary $n$-variate circuit of multiplicative complexity $s$, we conclude that $\mathcal{G}_{\sqrt{n},\sqrt{n},r-1}(Y,Z)$ hits all polynomials in the closure of $n$-variate circuits of multiplicative complexity $s$.
	Finally, note that the definition of $\mathcal{G}_{\sqrt{n},\sqrt{n},r-1}(Y,Z)$ immediately implies the claimed bounds on the seed length and degree of the generator.
\end{proof}

Combining \autoref{thm:hsg} with \autoref{thm:brank lb}, we obtain the following corollary.
To the best of our knowledge, this is the first non-trivial hitting set generator for circuits of multiplicative complexity $s \le o(n)$.

\begin{corollary}
	There is an explicit hitting set generator of seed length $(8 \sqrt{3} + o(1)) \sqrt{ns}$ that hits the closure of $n$-variate circuits of multiplicative complexity $s$.
\end{corollary}

One can also state \autoref{thm:hsg} as a win-win result: either there are extremely fast algorithms for $n \times n \times n$ matrix multiplication, or there is a non-trivial deterministic algorithm for testing polynomial identities given by small circuits.

\begin{corollary}
	Let $\F$ be a field of characteristic zero and let $\omega$ denote the exponent of matrix multiplication over $\F$.
	At least one of the following is true.
	\begin{enumerate}
		\item
			$\omega = 2$.
		\item
			For any positive constants $\eps, \delta > 0$ that satisfy $2 \omega \eps + 2 \delta < \omega - 2$, there is an explicit hitting set generator of seed length $O(n^{1 - \eps})$ that hits $n$-variate algebraic circuits of multiplicative complexity $O(n^{1 + \delta})$.
			If these circuits are also restricted to have degree $n^{O(1)}$ and size $n^{O(1)}$, then this yields a deterministic algorithm to test identities given by such circuits that runs in $\exp(O(n^{1 - \eps} \log n))$ time.
	\end{enumerate}
\end{corollary}

\paragraph{Acknowledgments}

We thank Shubhang Kulkarni for telling us about the work of \textcite{BP87}.
We also thank the anonymous reviewers for comments that helped improve the presentation of this work.

\printbibliography

\end{document}